\documentclass[11pt]{article}
\usepackage[draft]{ktmacros}
\usepackage{fullpage}

\newcommand{\val}{\texttt{val}}
\newcommand{\ATmax}{\ensuremath{\mathcal{A}^T_{\max}}}

\newcommand{\Alt}{\ensuremath{\mathcal{A}_{LT}}}

\newcommand{\A}{\ensuremath{\mathcal{A}}}
\newcommand{\M}{\ensuremath{\mathcal{M}}}

\setcitestyle{square,comma,numbers}

\newif\ifconf
\conffalse

\title{Privacy-Computation trade-offs in Private Repetition and Metaselection}
\author{Kunal Talwar\footnote{\texttt{ktalwar@apple.com}}\\Apple\\}
\date{\vspace{-5ex}}

\begin{document}

\maketitle

\begin{abstract}
A Private Repetition algorithm takes as input a differentially private algorithm with constant success probability and boosts it to one that succeeds with high probability. These algorithms are closely related to private metaselection algorithms that compete with the best of many private algorithms, and private hyperparameter tuning algorithms that compete with the best hyperparameter settings for a private learning algorithm. Existing algorithms for these tasks pay either a large overhead in privacy cost, or a large overhead in computational cost. In this work, we show strong lower bounds for problems of this kind, showing in particular that for any algorithm that preserves  the privacy cost up to a constant factor, the failure probability can only fall polynomially in the computational overhead. This is in stark contrast with the non-private setting, where the failure probability falls exponentially in the computational overhead. By carefully combining existing algorithms for metaselection, we prove computation-privacy tradeoffs that nearly match our lower bounds.
\end{abstract}

\section{Introduction}

Randomized algorithms have probabilistic guarantees. Often one designs an algorithm that succeeds (for an appropriate notion of success) with constant probability.  A standard approach to boosting the success probability of a randomized algorithm is to run the algorithm with fresh randomness multiple times and take the best of the results of individual runs. Standard tail inequalities then allow for proving bounds on the success probability of the resulting algorithm. Suppose that an algorithm $\mathcal{A}$ produces an output with quality score at least $m$ with probability at least $\frac{1}{2}$.\footnote{Without loss of generality, we assume our goal is to maximize a quality score; our results can be translated immediately to a setting where want to minimize a quality score. The constant $\frac 1 2$ can be replaced by any other constant.} As an example, the algorithm of~\cite{GoemansW95} for finding a maximum cut in the graph finds a cut with value within a small constant factor of the optimum with constant probability. Let $\ATmax$ be the algorithm that runs $\mathcal{A}$ $T$ times to get outputs $o_1, \ldots, o_T$ and releases the $o_i$ with the highest quality score.
The basic \emph{repetition theorem} says that $\ATmax$ produces an output with quality score at least $m$ except with probability $\frac{1}{2^T}$. An important aspect of this result is the relationship between the failure probability and the number of repetitions: logarithmically many repetitions suffice to make the failure probability polynomially small.  In the example above, repeating the maximum cut algorithm 20 times ensures that the best cut amongst these is approximately optimal except with probability $10^{-6}$. Various versions of repetition theorems have been studied in literature, e.g. improving on the amount of randomness needed~\citep{HooryLW06}, or allowing for parallelizability in case of multiple round protocols~\citep{Raz95, Holenstein07}.

In this work, we are interested in this question when the algorithm $\mathcal{A}$ is differentially private (see \cref{sec:preliminaries} for a precise definition). We assume we have a differentially private algorithm that outputs an object (e.g. an ML model) and a score (e.g. approximate accuracy on a \emph{test} dataset), where the score is ``high'' with constant probability. Private Repetition theorems allow for boosting the success probability, while controlling the privacy cost. A simple repetition theorem uses the $\ATmax$ above. To get failure probability below $\gamma$, it would suffice to set $T=\log_2 \frac 1 \gamma$. This increases the privacy parameters (e.g. the $\eps$ in $(\eps,\delta)$-DP) of the algorithm by a factor of $T$ which may be undesirable. \ifconf Liu and Talwar~\cite{LiuT19} \else\citet{LiuT19}\fi  designed and analyzed a different algorithm for private repetition that we will refer to as \Alt\  for the rest of introduction.  \Alt\  allows for a constant ($3\times$) increase in the privacy cost, while allowing an arbitrarily small failure probability $\gamma$. However, it requires running the algorithm $\mathcal{A}$ about $O(\frac{1}{\gamma})$ times, instead of the logarithmic dependence on $\gamma^{-1}$ in the naive repetition theorem.

\begin{figure}[t]
  \begin{subfigure}{0.3\textwidth}
    \centering

\begin{tikzpicture}

  \fill[dashed, gray!30] (0,0)--(4,0)--(4,0.5)--(0,0.5)--cycle;

  \fill[dashed, gray!30] (0,0)--(0.5,0)--(0.5,4)--(0,4)--cycle;

  \draw[->] (0,0)--(4.3,0);
  \node[below] at (2,-0.5) {Privacy Overhead $c$};
  \draw[->] (0,0)--(0,4.3);
  \node[above, rotate=90] at (-1.3,2) {Computational Overhead $T$};

  \draw[dashed, gray] (0,0.5)--(4,0.5);
  \draw[dashed, gray] (0.5,0)--(0.5,4);


  \draw[-,very thin] (-0.05,0.5)--(0.05,0.5);
  \node[left] at (0,0.5) {\small{$\log_2 \frac 1 \gamma$}};
  \draw[-,very thin] (-0.05,3)--(0.05,3);
  \node[left] at (0,3) {\small{$\frac 1 {\gamma}$}};
  \draw[-,very thin] (-0.05,2.5)--(0.05,2.5);
  \node[left] at (0,2.5) {\small{$\frac 1 {\gamma^{O(1)}}$}};
  \draw[-,very thin] (0.5,-0.05)--(0.5,0.05);
  \node[below] at (0.5,0) {\small{$\Theta(1)$}};
  \draw[-, very thin] (3,0.05)--(3,-0.05);
  \node[below] at (3.3,0) {\tiny{$\log_2 \gamma^{-1}$}};
  \draw[-, very thin] (2.5,0.05)--(2.5,-0.05);
  \node[below] at (1.9,0) {\tiny{$O(\tfrac{\log\gamma^{-1}}{\log\log \gamma^{-1}})$}};

  \draw[fill=green] (3,0.5) circle (0.1cm);
  \draw[->, green] (3.6,1) -- (3.2,0.7);
  \node[right] at (3.6,1) {$\ATmax$};

  \draw[fill=green] (0.5, 3) circle (0.1cm);
  \draw[->, green] (1,3.6) -- (0.7, 3.2);
  \node[above] at (1,3.6) {\Alt};

\draw[smooth, tension=5, blue, dashed] plot (3,0.5) -- (2.2,1) -- (1.9,1.2) -- (1.52,1.52) -- (1.2,1.9) -- (1,2.2)  -- (0.5, 3);
\draw[pattern=north east lines, pattern color=red]  plot coordinates {(2.5,0.5) (1.7,1) (1.4,1.2) (1.2,1.4) (1,1.7) (0.5,2.5) (0.5,0.5) (2.5,0.5)};
\draw[red] plot[smooth, tension=1] (2.5,0.5) -- (1.7,1) -- (1.4,1.2)  -- (1.2,1.4) -- (1,1.7)  -- (0.5,2.5);
 \draw[->, blue] (2.0,1.9) -- (1.6, 1.6);
  \node[right] at (2.0, 1.9)  {\small{Hybrid (Thm. \ref{thm:LTmax})}};
 \draw[->, red] (1.3,2.6) -- (0.9, 2);
  \node[right] at (1.3,2.6) {\small{LB (Thm. \ref{thm:main_lb_coded})}};

  \end{tikzpicture}
  \end{subfigure}
  \hspace{0.15\textwidth}
   \begin{subfigure}{0.4\textwidth}
    \centering

\begin{tikzpicture}

  \fill[dashed, gray!30] (0,0)--(5,0)--(5,0.5)--(0,0.5)--cycle;

  \fill[dashed, gray!30] (0,0)--(0.5,0)--(0.5,4)--(0,4)--cycle;

  \draw[->] (0,0)--(5.3,0);
  \node[below] at (2,-0.5) {Privacy Overhead $c$};
  \draw[->] (0,0)--(0,4.3);
  \node[above, rotate=90] at (-1.3,2) {Computational Overhead $T$};

  \draw[dashed, gray] (0,0.5)--(5,0.5);
  \draw[dashed, gray] (0.5,0)--(0.5,4);

  \draw[-,very thin] (-0.05,0.5)--(0.05,0.5);
  \node[left] at (0,0.5) {\small{$K\log_2 \frac 1 \gamma$}};
  \draw[-,very thin] (-0.05,3)--(0.05,3);
  \node[left] at (0,3) {\small{$\frac K {\gamma}$}};
  \draw[-,very thin] (-0.05,2.5)--(0.05,2.5);
  \node[left] at (0,2.5) {\small{$\frac K {\gamma^{O(1)}}$}};
  \draw[-,very thin] (0.5,-0.05)--(0.5,0.05);
  \node[below] at (0.5,0) {\small{$\Theta(1)$}};
  \draw[-, very thin] (3,0.05)--(3,-0.05);
  \node[below] at (3.3,0) {\tiny{$\log_2 \gamma^{-1}$}};
  \draw[-, very thin] (2.5,0.05)--(2.5,-0.05);
  \node[below] at (1.9,0) {\tiny{$O(\tfrac{\log\gamma^{-1}}{\log\log \gamma^{-1}})$}};
  \draw[-, very thin] (4.5,0.05)--(4.5,-0.05);
  \node[below] at (4.8,0) {\small{$O(K \log \frac 1 \gamma)$}};

  \draw[fill=green] (4.5,0.5) circle (0.1cm);
  \draw[->, green] (5.2,1) -- (4.7,0.7);
  \node[right] at (5.2,1) {$\ATmax$};

  \draw[fill=green] (0.5, 3) circle (0.1cm);
  \draw[->, green] (1,3.6) -- (0.7, 3.2);
  \node[above] at (1,3.6) {\Alt};

  \draw[fill=blue] (3,0.5) circle (0.1cm);

\draw[smooth, tension=5, blue, dashed] plot (3,0.5) -- (2.2,1) -- (1.9,1.2) -- (1.52,1.52) -- (1.2,1.9) -- (1,2.2)  -- (0.5, 3);
\draw[pattern=north east lines, pattern color=red]  plot coordinates {(2.5,0.5) (1.7,1) (1.4,1.2) (1.2,1.4) (1,1.7) (0.5,2.5) (0.5,0.5) (2.5,0.5)};
\draw[smooth, tension=5, red] plot (2.5,0.5) -- (1.7,1) -- (1.4,1.2)  -- (1.2,1.4) -- (1,1.7)  -- (0.5,2.5);
 \draw[->, blue] (2.7,1.7) -- (1.6, 1.6);
 \draw[->, blue] (2.7,1.7) -- (3,0.7);
  \node[above] at (3.2, 1.7)  {\small{Hybrid (Thm. \ref{thm:LTmax})}};
 \draw[->, red] (1.3,2.6) -- (0.9, 2);
  \node[right] at (1.3,2.6) {\small{LB (Thm. \ref{thm:main_lb_hyperparameter})}};

  \end{tikzpicture}
        \end{subfigure}
  \caption{Existing and new results on Private Repetition (left) and Private Hyperparameter tuning and Metaselection (right). Green dots show the existing upper bounds $\Alt$ and $\ATmax$ discussed above, and the gray regions depict the existing excluded regions. The red line shows our new lower bounds: we exclude the full region below the red line. The lower bound nearly matches the blue dashed line that depicts our analysis of hybrid algorithms.}\label{fig:results}
\end{figure}
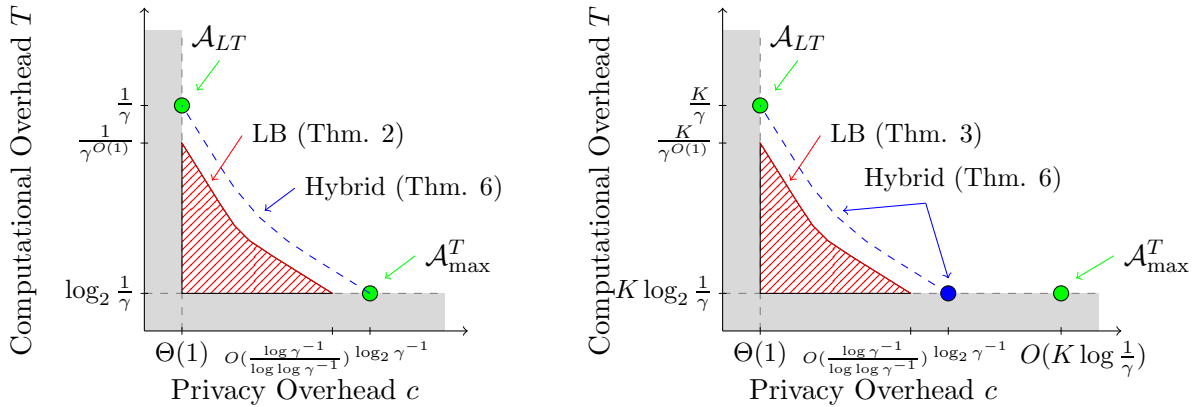

In other words, existing private repetition theorems either pay a $\log \frac 1 \gamma$ overhead in the privacy cost, or a $\poly(\frac 1 \gamma)$ overhead in the run time. A similar dichotomy exists for private hyperparameter selection, where one wants to privately select the best model amongst those resulting from different hyperparameter choices. Indeed private hyperparameter tuning is one of the standard uses of private repetition theorems ~\citep{LiuT19, PapernotS22} and the \Alt\  algorithm has been used in practice recently by Israel’s Ministry of Health in their release of 2014 live births microdata~\citep{HodC24}. The private learning algorithms for a single hyperparameter setting can can be computationally expensive to run.
As an example, in \ifconf Hod and Canetti~\cite{HodC24}\else\citet{HodC24}\fi, the algorithm $\mathcal{A}$ can involve training a CTGAN~\citep{XuSCV19,YJdS19,rosenblatt2020differentially} using DPSGD~\citep{DLDP} or PATE~\citep{PapernotAEGT17,PapernotSMRTE18} for a certain set of hyperparameters. In such cases, the $\poly(\frac 1 \gamma)$ run time overhead can be quite significant.

As we show in~\cref{sec:tradeoffs}, the two algorithms above can be combined to give a smooth trade-off between the privacy cost overhead and the computational overhead. It is natural to ask however if such a trade-off is necessary: is there a way to keep the privacy overhead to constant while paying only logarithmically in the run time?

In this work, we give a negative answer to this question. We show that for any algorithm with constant privacy overhead, the run time must necessarily be polynomial in $\frac{1}{\gamma}$. More generally, our lower bound shows that the trade-off between the computational overhead  (denoted by $T$ below) and the privacy overhead (denoted by $c$ below) is nearly tight.

\begin{theorem*} (Informal version of \cref{thm:main_lb_coded})
Let $\mathcal{A}$ be an oracle algorithm that satisfies the following properties:
\begin{description}
    \item{Privacy:} For any $\eps$-DP mechanism $\M : D^n \rightarrow (R \times \Re)$, $\mathcal{A}^{\M}$ is $(c\eps,\delta)$-DP.
    \item{Utility:} For any input $d$, let $m$ be such that $\Pr[\val(\M(d)) \geq m] \geq \frac 1 2$. Then $\Pr[\val(\mathcal{A}^\M(d)) \geq m] \geq 1-\gamma$.
    \item{Runtime:} $\mathcal{A}$ on input $d$ makes at most $T$ calls to $\mathcal{M}$ (on any inputs) in expectation.
\end{description}
Then for $\delta < O(\frac {\eps\gamma} {\ln T})$, $T \geq \gamma^{\Omega(1/c)}$, or equivalently, $c \geq \Omega(\frac{\ln \frac 1 {\gamma}}{\ln T})$.
\end{theorem*}
These results are shown pictorially in~\cref{fig:results} (left). The two green filled dots denote the trade-offs achieved by \Alt\  and \ATmax. The blue dashed line is the upper bound achieved by combining these algorithms (\cref{sec:tradeoffs}). The grey shaded region can be excluded by straightforward arguments. The lower bound above excludes the red striped region. In particular, it shows that for any constant $c$, $T$ must be at least $\poly(\frac{1}{\gamma})$.
At the other extreme, any $T$ that is polylogarithmic in $\frac 1 \gamma$ requires $c \geq \Omega(\frac{\ln \frac 1 {\gamma}}{\ln \ln \frac 1 \gamma})$. These results hold even when the target value $m$ is public.

We extend our results to the setting where the target value $m$ is achieved by the mechanism $\M$ with a probability that is different from $\frac{1}{2}$. More generally, in the hyperparameter tuning setting, there are $K$ possible settings of hyperparameters, and our utility bound asks that the algorithm be competitive with the median value of the best hyperparameter setting.
\begin{theorem*} (Informal version of~\cref{thm:main_lb_hyperparameter})
Let $\mathcal{A}$ be an oracle algorithm that satisfies the following properties:
\begin{description}
    \item{Privacy:} For any $\eps$-DP mechanism $\M : D^n \times [K] \rightarrow (R \times \Re)$, $\mathcal{A}^{\M} : D^n \rightarrow (R \times \Re)$ is $(c\eps,\delta)$-DP.
    \item{Utility:} For any input $d$, let $m$ be such that $\Pr_{\M}\Pr_{k \sim [K]}[\val(\M(d,k)) \geq m] \geq \frac{1}{2K}$. Then $\Pr[\val(\mathcal{A}^\M(d)) \geq m] \geq 1-\gamma$.
    \item{Runtime:} $\mathcal{A}$ on input $d$ makes at most $T$ calls to $\mathcal{M}$ (on any inputs) in expectation.
\end{description}
Then for $\delta < O(\frac {\eps\gamma} {\ln T})$, $T \geq  K\gamma^{\Omega(\frac 1 c)}$, or equivalently, $c \geq \Omega(\frac{\ln \frac 1 {\gamma}}{\ln (T/K)})$.
\end{theorem*}
We show these results in ~\cref{fig:results} (right). As before, the two green filled dots denote the trade-offs achieved by \Alt\  and \ATmax. Note the here $\ATmax$ is significantly worse, and the better upper bounds is achieved by combining \Alt\  with $\ATmax$. As before the grey shaded region is excluded by straightforward arguments, and the new lower bound exclude the red striped region. In particular, it shows that for any constant $c$, $T$ must be at least $K \cdot \poly(\frac{1}{\gamma})$.
At the other extreme, making $T/K$ polylogarithmic in $\frac 1 \gamma$ requires $c \geq \Omega(\frac{\ln \frac 1 {\gamma}}{\ln \ln \frac 1 \gamma})$. As before, the lower bound holds for known target value $m$.

This has strong implications for private hyperparameter tuning. Indeed our results imply that any private hyperparameter tuning algorithm that boosts the success probability to $(1-\gamma)$ while paying only a constant overhead in privacy cost must pay a $K\cdot\poly(\gamma^{-1})$ computational overhead.

We remark that both the upper bounds translate $\eps$-DP algorithms to $c\eps$-DP algorithms, and $(\eps,\delta)$-DP algorithms to $(c\eps, T\delta)$-DP algorithms; similar results hold for Renyi DP algorithms as well~\citep{PapernotS22}. Our lower bounds apply even when the input algorithm is $\eps$-DP, and the output algorithm is allowed to be $(\eps,\delta)$-DP  (or Renyi DP).

\subparagraph*{Other Related Work:} \ifconf Gupta \emph{et al.}~\cite{GuptaLMRT10} \else\citet{GuptaLMRT10}\fi first proved a repetition theorem for the case when the target threshold is known and the value is a low-sensitivity function. Their result required $T=O(1/\gamma^2)$ for repetition and $T=O(K^2/\gamma^2)$ for metaselection. As discussed, the random stopping approach of \citep{LiuT19} improves these to $T=O(1/\gamma)$ and $T=O(K/\gamma)$ without any restrictions. Subsequently, ~\cite{PapernotS22} studied these questions under Renyi Differential privacy, and showed positive results for a variety of distributions of stopping times.

\subparagraph*{Decision vs. Optimization:} Repetition theorems in complexity theory are normally stated for decision problems. This is justified by the fact that one can reduce optimization problems to decision problems with logarithmic (in the range) overhead. However a logarithmic overhead in privacy cost would be quite significant as one typically wants $\eps$ to be a small constant. For this reason, all the above works~\citep{GuptaLMRT10, LiuT19, PapernotS22} directly address optimization problems. As decision problems can be reduced to optimization, our results also apply to decision problems.

\subparagraph*{The need for Private Repetition:} Private algorithms come in many flavors. In many cases, the algorithm itself or a slight modification can be better analyzed to make its failure probability small, at a small cost to the utility criteria. For example, the Laplace mechanism~\cite{DworkMNS06j} and the exponential mechanism~\citep{McsherryT07} can directly give utility bounds that hold with probability $(1-\gamma)$ for any $\gamma$. In these cases, we do not need private repetition algorithms that use a base algorithm as an oracle. However, for many other private algorithms, we do no have such sliding scale guarantees and one can only control the expected utility, or median utility. This is often the case when the randomization plays a role in ``non-private'' part of the algorithm, e.g. in the use of tree embeddings~\cite{GuptaLMRT10, FeldmanMST24}, hashing~\cite{BassilyNST20} or other more custom analyses~\cite{GuptaLMRT10, DinitzKLV25}. Finally, for some problems, (private) algorithms work well in practice but may not have theoretical guarantees (or work much better than their theoretical guarantees) and one does not have a knob to tune the failure probability. For these two kinds of algorithms, private repetition theorems are an invaluable tool for reducing the failure probability.

Metaselection and hyperparameter tuning algorithms are even more broadly useful. It is fairly common to have algorithms that take a hyperparameter as input and work well either provably or empirically. As an example, algorithms in the Propose-test-release framework~\citep{DworkL09} work well when the proposal is true for the distribution. In cases when the proposal can be parameterized (e.g. in~\citep{BrownHHKLOPS24}), it becomes a hyperparameter. Practical private optimization algorithms (e.g.~\citep{DLDP}) typically involve multiple hyperparameters. Often the hyperparameters cannot be well-tuned on public proxy datasets, and can be a source of leakage of private information~\cite{PapernotS22}. Private Hyperparameter tuning algorithms are practically used in such settings and the computational cost can be a significant concern.

\section{Preliminaries}
\label{sec:preliminaries}

Differential Privacy~\citep{DworkMNS06j} constrains the distribution of outcomes on neighboring datasets to be close. The Hamming distance $|d-d'|_H$ between two datasets $d,d' \in \mathcal{D}^n$ is the number of entries in which they differ. For our purposes, two datasets $d, d' \in \mathcal{D}^n$ are \emph{neighboring} if they differ in one entry, i.e. if their Hamming distance is 1.
\begin{definition}
    A randomized algorithm $\M : \mathcal{D}^n \rightarrow \mathcal{R}$ is $(\eps,\delta)$-differentially private (DP) if for any pair of neighboring datasets $d, d'$ and any $S \subseteq \mathcal{R}$,
    \begin{align*}
        \Pr[\M(d) \in S] & \leq e^{\eps} \cdot \Pr[\M(d') \in S] + \delta.
    \end{align*}
    We abbreviate $(\eps,0)$-DP as $\eps$-DP and refer to it as pure differential privacy.
\end{definition}
Sometimes it is convenient to define a mechanism on a subset of inputs, and establish its privacy properties on this restricted subset. The following extension lemma due to \ifconf Borgs \emph{et al.}~\cite{BorgsCSZ18}\else ~\citet{BorgsCSZ18}\fi shows that any such partial specification can be extended to the set of all datasets at a small increase in the privacy cost.
\begin{proposition}[Extension Lemma]
\label{prop:extension}
Let $\M_B$ be an $\eps$-differentially private algorithm restricted to $B \subseteq D^n$ so that for any $d, d' \in B$, and any $S \subseteq R$, $\frac{\Pr[\M_B(d) \in S]}{\Pr[\M_B(d') \in S]} \leq \exp(\eps \cdot |d-d'|_H)$. Then there exists a randomized algorithm $\M$ defined on the whole input space $D$ which is $2\eps$-differentially private and satisfies that for every $d \in B$, $\M(d)$ has the same distribution as $\M_B(d)$.
\end{proposition}
We will be interested in private mechanisms that output a tuple $(o,v) \in \mathcal{R} \subset \mathcal{R}_o \times \mathbb{R}$ where $\mathcal{R}_o$ is some arbitrary output space. For example, if the mechanism is computing an approximate maximum cut in a graph, $o$ could be a subset of vertices of a graph and $v$ the (approximate) number of edges in the input graph that leave this subset $o$. In the case of hyperparameter tuning, $o$ would be a model and $v$ an estimate of its accuracy. We denote by $\val(\M(d))$ the real value $v$ when $\M(d) = (o,v)$. For a mechanism $\M : D^n \rightarrow \mathcal{R}_o \times \mathbb{R}$, we let $Median(\M(d)) = \sup\{z : \Pr_{(o,v) \sim \M(d)}[v \geq z] \geq \frac 1 2\}$ denote the median of $v$ when $(o,v) \sim \M(d)$. A private amplification algorithm takes as input a dataset $d$ and a failure probability $\gamma$, and given oracle access to a mechanism $\M$, outputs an $(o,v) \in \mathcal{R}$ such that $v \geq Median(\M(d))$ with probability $(1-\gamma)$. Here this probability is taken over both the internal randomness of the amplification algorithm as well as the randomness in $\M$. A private amplification algorithm may need to make $T$ oracle calls to $\M$, and may be $c\eps$-DP whenever $\M$ is $\eps$-DP. Our goal is to understand the trade-off between the computational overhead $T$ and the privacy overhead $c$.

We will assume that $\mathcal{A}$ can only output an $(o,v)$ tuple that is produced by a run of $\M$ on some input. This is a natural restriction in all settings, and is needed to make the problem meaningful.\footnote{E.g., a trivial algorithm that computes $(o,v) \sim \M(d)$ and outputs  $(o,\infty)$ would not be useful in any application.} Thus in this paper, we will always make the assumption that the oracle algorithm $\mathcal{A}$ selects a tuple $(o,v)$ from the outputs it receives from calls to $\M$.

A private metaselection algorithm operates on a set of $K$ private algorithms $\M_1, \ldots, \M_K$, where each $\M_i : \mathcal{D}^n \rightarrow \mathcal{R} \subset \mathcal{R}_o \times \mathbb{R}$. It takes as input a dataset $d$ and a failure probability $\gamma$, and given oracle access to mechanisms $\M_i$, outputs an $(o,v) \in \mathcal{R}$ such that $v \geq \max_i Median(\M_i(d))$ with probability $(1-\gamma)$. As above, the probability is taken over both the internal randomness of the metaselection algorithm as well as the randomness in $\M_i$'s.  The metaselection algorithm may need to make $T$ oracle calls to the $\M_i$'s, and may be $c\eps$-DP whenever each $\M_i$ is $\eps$-DP. Our goal as before is to understand the trade-off between the computational overhead $T$ and the privacy overhead $c$. Note that hyperparameter tuning can be phrased as metaselection, by treating the private training algorithm for each setting of hyperaprameters as a separate algorithm amongst $K$ options. More generally, we can phrase metaselection and hyperparameter tuning as special cases of repetition, where the target value is the $(1-\frac{1}{2K})$\textit{th} quantile of the distribution, instead of the median. This variant can handle the case when the set of hyperparameter settings may be large (or even unbounded) but we expect a non-trivial measure of the hyperparameter settings to be good.

\begin{proposition}
Let $\beta > 0$ and $\mathcal{A}$ be an oracle algorithm that satisfies the following properties:
\begin{description}
    \item{Privacy:} For any $\eps$-DP mechanism $\M : D^n \rightarrow (R \times \Re)$, $\mathcal{A}^{\M}$ is $(c\eps,\delta)$-DP.
    \item{Utility:} For any input $d$, let $m$ be such that $\Pr[\val(\M(d)) \geq m] \geq \beta$. Then $\Pr[\val(\mathcal{A}^\M(d)) \geq m] \geq 1-\gamma$.
    \item{Runtime:} $\mathcal{A}$ on input $d$ makes at most $T = T(c, \beta)$ calls to $\M$ (on any inputs) in expectation.
\end{description}
Let $H$ be a hyperparameter space, and let $\mu_H$ be a measure on $H$. Then there is an oracle algorithm $\tilde{\mathcal{A}}$ that satisfies the following properties:
\begin{description}
    \item{Privacy:} For any $\eps$-DP mechanism $\M : D^n \times H \rightarrow (R \times \Re)$, $\tilde{\mathcal{A}}^{\M}$ is $(c\eps,\delta)$-DP.
    \item{Utility:} For any input $d$, let $m$ be such that $\Pr_{h \sim \mu_H}[\val(\M(d)) \geq m] \geq \beta$. Then $\Pr[\val(\tilde{\mathcal{A}}^\M(d)) \geq m] \geq 1-\gamma$.
    \item{Runtime:} $\tilde{\mathcal{A}}$ on input $d$ makes at most $T = T(c, \beta)$ calls to $\M$ (on any inputs) in expectation.
\end{description}
In particular, when $H=[K]$ and $\mu_H$ is uniform, then $\tilde{\mathcal{A}}$ competes with the largest median value $\max_{k\in[K]} Median(\val(\M(d,k)))$ and makes $T=T(c, \frac{1}{2K})$ calls to $\M$.
\end{proposition}
\begin{proof}
    Define $\M'$ that on input $d$, samples a random $h \sim \mu_H$ and runs $\M(d,h)$. Running $\mathcal{A}$ with $\M'$ gives the first result. The second result uses the fact that the output of $\M'$ when $H=[K]$ and $\mu$ is uniform is at least the largest median with probability  at least $\frac 1 {2K}$.
\end{proof}

Differential privacy constrains the likelihood of events on neighboring datasets. This also implies constraints for datasets at some bounded distance. The following forms of these constraints will be useful. The proofs of the following consequences of group privacy~\citep{DworkR14} are elementary.

\begin{proposition}
\label{prop:puredp_probs}
    Let $\M$ satisfy $\eps$-DP. Then for datasets $d, d'$ such that $|d-d'|_H \leq \Delta$, and any event $E$,
    \begin{align*}
        \Pr[\M(d) \in E] &\leq e^{\eps\Delta}\Pr[\M(d') \in E].
    \end{align*}
    In particular, if $\Pr[\M(d) \in E] \geq \frac 1 4$, then $\Pr[\M(d') \in E] \geq e^{-\eps\Delta}/4$.
\end{proposition}

\begin{proposition}
\label{prop:ed_dp_probs}
    Let $\M$ satisfy $(\eps, \delta)$-DP. Then for data sets $d,d'$ such that $|d-d'|_H = \Delta$, and any event $E$,
    \begin{align*}
        \Pr[\M(d)\in E] &\leq e^{\eps\Delta}(\Pr[\M(d') \in E] + \Delta\delta).
    \end{align*}
    In particular, if $\delta < e^{-\eps\Delta}/8\Delta$ and $\Pr[\M(d) \in E] \geq \frac 1 4$, then $\Pr[\M(d') \in E] \geq e^{-\eps\Delta}/8$.
\end{proposition}

\noindent A similar statement holds for Renyi differential privacy. 
\begin{proposition}
\label{prop:rdp_probs}
    Let $\M$ satisfy $(\alpha, \eps)$-Renyi DP for $\alpha>1$. Then for data sets $d,d'$ such that $|d-d'|_H = \Delta$, and any event $E$,
    \begin{align*}
        \Pr[\M(d)\in E] &\leq e^{\eps\sum_{i=1}^{\Delta}(1 - \frac 1 \alpha)^i} \cdot (\Pr[\M(d') \in E])^{(1-\frac{1}{\alpha})^{\Delta}}.
    \end{align*}
In particular, for $\alpha \geq \Delta+1$, we get
    \begin{align*}
        \Pr[\M(d)\in E] &\leq e^{(\alpha-1)\eps} \cdot (\Pr[\M(d') \in E])^{1/e}.
    \end{align*}
    It follows that for $\alpha \geq \Delta+1$, if $\Pr[\M(d)\in E] \geq \frac 1 4$, then $\Pr[\M(d') \in E] \geq e^{-e(\alpha-1)\eps}/44$.
\end{proposition}
\begin{proof}
Let $d'=d_0,d_1,\ldots,d_{\Delta}=d$ be a sequence of datasets where $d_i$ and $d_{i+1}$ are adjacent. A consequence of $(\alpha,\eps)$-RDP ~\citep[Prop. 10]{Mironov17} is that for any event $E$,
\begin{align*}
    \Pr[\M(d_{i+1}) \in E] &\leq e^{\eps(1-\frac 1 \alpha)} \cdot (\Pr[\M(d_{i}) \in E])^{(1-\frac 1 \alpha)}.
\end{align*}
This implies the base case ($k=1$) of the claim that for all $k$,
\begin{align*}
    \Pr[\M(d_k) \in E] &\leq e^{\eps\sum_{i=1}^{k}(1 - \frac 1 \alpha)^i} \cdot (\Pr[\M(d_0) \in E])^{(1-\frac{1}{\alpha})^{k}}.
\end{align*}
Suppose that the claim holds for $d_{k-1}$. We now inductively prove it for $d_k$. We write
\begin{align*}
    \Pr[\M(d_{k}) \in E] &\leq e^{\eps(1-\frac 1 \alpha)} \cdot (\Pr[\M(d_{k-1}) \in E])^{(1-\frac 1 \alpha)}\\
    &\leq e^{\eps(1-\frac 1 \alpha)} \cdot \left(e^{\eps\sum_{i=1}^{k-1}(1 - \frac 1 \alpha)^i} \cdot (\Pr[\M(d_0) \in E])^{(1-\frac{1}{\alpha})^{k-1}}\right)^{(1-\frac 1 \alpha)}\\
    &= e^{\eps\sum_{i=1}^{k}(1 - \frac 1 \alpha)^i} \cdot (\Pr[\M(d_0) \in E])^{(1-\frac{1}{\alpha})^{k}}.
\end{align*}
This completes the proof of the first part. For the second part, we use the following two facts. Firstly, for $\alpha \geq \Delta+1$, $(1-\frac 1 \alpha)^\Delta > \frac 1 e$, so that the exponent of $(\Pr[\M(d_0) \in E])$ is at least $(1/e)$ (when $k=\Delta$). Secondly, the geometric series $\sum_{i=1}^{\Delta}(1 - \frac 1 \alpha)^i \leq \sum_{i=1}^{\infty}(1 - \frac 1 \alpha)^i = \alpha -1$. The third part follows immediately by rearrangement.
\end{proof}
\section{Main Lower Bound}

\begin{theorem}
\label{thm:main_lb}
Let $\mathcal{A}$ be an oracle algorithm that satisfies the following properties:
\begin{description}
    \item{Privacy:} For any $\eps$-DP mechanism $\M : D^n \rightarrow (R \times \Re)$, $\mathcal{A}^{\M}$ is $(c\eps,\delta)$-DP.
    \item{Utility:} For any input $d$, $\Pr[\val(\mathcal{A}^\M(d)) \geq Median(\val(\M(d)))] \geq 1-\gamma$.
    \item{Runtime:} $\mathcal{A}$ on input $d$ makes at most $T$ calls to $\M$ on input $d$ in expectation for some $T \geq e^{\eps}$.
\end{description}
Then for $\delta < \frac {\eps\gamma} {4\ln 4T}$, $T \geq (8\gamma)^{-\frac{1}{4c}} / 4$, or equivalently, $c \geq \frac{\ln \frac 1 {8\gamma}}{2\ln 4T}$.
\end{theorem}
\begin{proof}
With some foresight, we set $\Delta = \lceil \frac{\ln 4T}{\eps}\rceil$, and set $q = e^{-\eps\Delta} \leq \frac 1 {4T}$. Consider datasets in  $\{0,1\}^\Delta$ where $d_0 = 0^\Delta$, and $d_1 = 1^\Delta$. Now we define a mechanism $\M$ such that
\begin{align*}
    \M(d_1) = \left\{ \begin{array}{ll} (r, 1) & \mbox{ w.p. } 1- q\\ (r', 0) & \mbox{ w.p. } q\end{array}\right. &\;\;\mbox{ and }\;\; \M(d_0) = \left\{\begin{array}{ll} (r, 1) & \mbox{ w.p. } q\\ (r', 0) & \mbox{ w.p. }1-q\end{array}\right. .
\end{align*}
Note that this specification on $\{d_0, d_1\}$ satisfies $\eps$-DP as $\frac{1-q}{q} \leq \frac 1 q  = e^{\eps |d_0 - d_1|_H}$. By the extension lemma (\cref{prop:extension}), this partial mechanism can be extended to a $2\eps$-DP mechanism on $\{0,1\}^\Delta$.

Consider a run of $\mathcal{A}^\M(d_0)$.  Let $E$ be the event that $\mathcal{A}^\M$ on input $d_0$ makes at most $2T$ calls to $\M(d_0)$. By Markov's inequality,  $\Pr[E] \geq \frac 1 2$.
Conditioned on this event $E$, the probability that $(r, 1)$ is the  output of \emph{any} of the at most $2T$ runs of $\M(d_0)$ is at most $2qT \leq \frac 1 2$. Since $\mathcal{A}^\M$ has never seen a $(r,1)$ output in this case, it follows that $Pr[\mathcal{A}^{\M}(d_0) = (r, 1) \mid E] \leq \frac 1 2$, so that $Pr[\mathcal{A}^{\M}(d_0) = (r', 0) \mid E] \geq \frac 1 2$. It follows that $Pr[\mathcal{A}^{\M}(d_0) = (r', 0)] \geq \frac 1 4$.

As $\M$ is $2\eps$-DP, the privacy property of $\mathcal{A}^{\M}$ implies that it is $(2c\eps, \delta)$-DP. By \cref{prop:ed_dp_probs}, it follows that for $\delta < e^{-2c\eps\Delta}/8\Delta$,
\begin{align*}
    \Pr[\mathcal{A}^\M(d_1) = (r',0)] &\geq e^{-2c\eps\Delta} / 8.
\end{align*}
By the utility property, the left hand side is at most $\gamma$. Since $e^{2c\eps\Delta} \leq e^{4c\ln 4T}$, it follows that
  $\gamma  \geq \frac{1}{8(4T)^{4c}}$.
Rearranging, we get the claimed result.
\end{proof}

\noindent By using \cref{prop:rdp_probs} in lieu of \cref{prop:ed_dp_probs}, we get a similar result for Renyi DP. We omit the straightforward proof, and similar corollaries of \cref{thm:main_lb_coded} and \cref{thm:main_lb_hyperparameter}.
\begin{corollary}
    In ~\cref{thm:main_lb}, if we replace the Privacy condition by
    \begin{description}
    \item{Privacy (RDP):} For any $\eps$-DP mechanism $\M : D^n \rightarrow (R \times \Re)$, $\mathcal{A}^{\M}$ is $(\alpha,c\eps)$-RDP.
    \end{description}
    Then if $\alpha \leq 1+ \lceil\frac {\ln T} {\eps}\rceil$, $T \geq (44\gamma)^{-\frac{1}{2ec}} / 4$, or equivalently, $c \geq \frac{\ln \frac 1 {44\gamma}}{2e\ln 4T}$.
\end{corollary}

\subsection{Allowing more general meta-algorithms}
One of the assumptions in \cref{thm:main_lb} is that when $\mathcal{A}^\M$ is run on input $d$, its oracle calls to $\M$ are also on input $d$. A more general oracle algorithm may, given input $d$, run $\M$ on additional inputs derived from $d$. We next show that the lower bound continues to hold, up to constants. At a high level, the lower bound in \cref{thm:main_lb} comes from the inability of $\mathcal{A}$ on input $d_0$ to see any sample that is unlikely in $\M(d_0)$ but likely in $\M(d_1)$. We show how to \emph{embed} such an instance in a mechanism over a slightly larger dataset, in a way where on input $\M(d_0)$, finding an output that is good for $\M(d_1)$ is still difficult.
\begin{theorem}
 \label{thm:main_lb_coded}
Let $\mathcal{A}$ be an oracle algorithm that satisfies the following properties:
\begin{description}
    \item{Privacy:} For any $\eps$-DP mechanism $\M : D^n \rightarrow (R \times \Re)$, $\mathcal{A}^{\M}$ is $(c\eps,\delta)$-DP.
    \item{Utility:} For any input $d$, $\Pr[\val(\mathcal{A}^\M(d)) \geq Median(\val(\M(d)))] \geq 1-\gamma$.
    \item{Runtime:} $\mathcal{A}$ on input $d$ makes at most $T$ calls to $\M$ (on any inputs) in expectation for some $T \geq e^{\eps}$.
\end{description}
Then for $\eps \leq 3$, $\delta < \frac {\eps\gamma} {40\ln 8T}$, $T \geq (8\gamma)^{-\frac{1}{40c}} / 8$, or equivalently, $c \geq \frac{\ln \frac 1 {8\gamma}}{40\ln 8T}$.
\end{theorem}
\begin{proof}
With some foresight, we set $\Delta = \lceil \frac{\ln 8T}{\eps}\rceil$, and set $q = e^{-\eps\Delta} \leq \frac 1 {8T}$.
Fix a vector $v \in \{0,1\}^{10\Delta}$, and let $B_v = \{u : |u-v|_H \leq \Delta\}$ be the radius $\Delta$ Hamming ball centered at $v$. We now define a mechanism $\M_v$ such that:
\begin{align*}
    \M_v(v) = \left\{ \begin{array}{ll} (r, 1) & \mbox{ w.p. } 1- q\\ (r', 0) & \mbox{ w.p. } q\end{array}\right. &\;\;\mbox{ and }\;\; \forall u \not\in B_v: \M_v(u) = \left\{\begin{array}{ll} (r, 1) & \mbox{ w.p. } q\\ (r', 0) & \mbox{ w.p. }1-q\end{array}\right. .
\end{align*}
Note that this specification on $\{v\} \cup B_v^c$ satisfies $\eps$-DP as $\frac{1-q}{q} \leq \frac 1 q  \leq e^{\eps |u-v|_H}$ for $u \not\in B_v$. By the extension lemma (\cref{prop:extension}), this partial mechanism can be extended to a $2\eps$-DP mechanism on $\{0,1\}^{10\Delta}$.

Consider a run of $\mathcal{A}^{\M_v}(d_0)$ where $d_0 = 0^{10\Delta}$.  Let $E$ be the event that $\mathcal{A}^{\M_v}$ on input $d_0$ makes at most $2T$ calls to $\M$ on any inputs. By Markov's inequality,  $\Pr[E] \geq \frac 1 2$.
For $v$ chosen uniformly at random, the probability that a single dataset $d$ chosen by $\mathcal{A}^{\M_v}(d_0)$ lies in $B_v$ is given by
\begin{equation*}
     \frac{|B_v|}{2^{10\Delta}} = \frac{{\binom{10\Delta}{\Delta}}}{2^{10\Delta}} \leq (\frac{10e}{1024})^\Delta \leq e^{-3\Delta} \leq \frac 1 {8T}.
\end{equation*}
Thus the likelihood of seeing an $(r,1)$ on any given call to $\M_v(d)$, when $v$ is chosen at random is
\begin{align*}
    \Pr[\M_v(d) = (r,1)] &\leq \Pr[d \in B_v] + \Pr[\M_v(d) = (r,1) \mid d \not\in B_v]\\
    &\leq \frac{1}{8T} + q \leq \frac{1}{4T}.
\end{align*}
Here and in the rest of the proof, the probability is taken over both the choice of $v$ and the randomness of $\mathcal{A}^{\M_v}$ and $\M_v$. Conditioned on the event $E$, the probability that $(r, 1)$ is the  output of \emph{any} of the at most $2T$ runs of $\M$ is at most $\frac{2T}{4T} \leq \frac 1 2$. Since $\mathcal{A}^\M$ has never seen an $(r,1)$ output in this case, it follows that $Pr[\mathcal{A}^{\M_v}(d_0) = (r, 1) \mid E] \leq \frac 1 2$, so that $Pr[\mathcal{A}^{\M}(d_0) = (r', 0) \mid E] \geq \frac 1 2$. It follows that $Pr[\mathcal{A}^{\M}(d_0) = (r', 0)] \geq \frac 1 4$.

The rest of the proof is now essentially identical to the earlier proof. As $\M_v$ is $2\eps$-DP, the privacy property of $\mathcal{A}^{\M_v}$ implies that it is $(2c\eps, \delta)$-DP. By \cref{prop:ed_dp_probs}, it follows that for $\delta < e^{-20c\eps\Delta}/80\Delta$,
\begin{align*}
    \Pr[\mathcal{A}^\M(d_1) = (r',0)] &\geq e^{-20c\eps\Delta} / 8.
\end{align*}

By the utility property, the left hand side is at most $\gamma$.
Since $e^{20c\eps\Delta} \leq e^{40c\ln 8T}$, it follows that
  $\gamma  \geq \frac{1}{8(8T)^{40c}}$.
Rearranging, we get the claimed result.
\end{proof}

\subsection{Lower Bounds for Hyperparameter Tuning}

We next show that the same essential arguments can be used to show a lower bound for private hyperparameter tuning.
\begin{theorem}
 \label{thm:main_lb_hyperparameter}
Let $\mathcal{A}$ be an oracle algorithm that satisfies the following properties:
\begin{description}
    \item{Privacy:} For any $\eps$-DP mechanism $\M : D^n \times [K] \rightarrow (R \times \Re)$, $\mathcal{A}^{\M} : D^n \rightarrow (R \times \Re)$ is $(c\eps,\delta)$-DP.
    \item{Utility:} For any input $d$, $\Pr[\val(\mathcal{A}^\M(d)) \geq max_{k \in [K]} Median(\val(\M(d, k)))] \geq 1-\gamma$.
    \item{Runtime:} $\mathcal{A}$ on input $d$ makes at most $TK$ calls to $\M$ (on any inputs) in expectation for some $T \geq e^{\eps}$.
\end{description}
Then for $\eps \leq 3$, $\delta < \frac {\eps\gamma} {40\ln 8T}$, $T \geq  (8\gamma)^{-\frac{1}{40c}} / 8$, or equivalently, $c \geq \frac{\ln \frac 1 {8\gamma}}{40\ln 8T}$.
\end{theorem}
\begin{proof}
We set $\Delta = \lceil \frac{\ln 8T}{\eps}\rceil$, and set $q = e^{-\eps\Delta} \leq \frac 1 {8T}$.
For a vector $v \in \{0,1\}^{10\Delta}$, and let $B_v = \{u : |u-v|_H \leq \Delta\}$ be the radius $\Delta$ Hamming ball centered at $v$. For a parameter setting $k \in [K]$, we now define a mechanism $\M_{v,k}$ such that:
\begin{align*}
    &\,\M_{v,k}(v,k) = \left\{ \begin{array}{ll} (r, 1) & \mbox{ w.p. } 1- q\\ (r', 0) & \mbox{ w.p. } q\end{array}\right. \;\;\mbox{, }\\ \forall u \not\in B_v: &\,\M_{v,k}(u,k) = \left\{\begin{array}{ll} (r, 1) & \mbox{ w.p. } q\\ (r', 0) & \mbox{ w.p. }1-q\end{array}\right. \\ \;\;\mbox{ and }\;\; \forall j \neq k, \forall u \in \{0,1\}^{10\Delta}: &\,\M_{v,k}(u,j) = (r',0) \mbox{ w.p. } 1.
\end{align*}
By the extension lemma (\cref{prop:extension}), for each $j \in [K]$, the partial mechanism $\M_{v,k}(\cdot, j)$ can be extended to a $2\eps$-DP mechanism on $\{0,1\}^{10\Delta}$.

Consider a run of $\mathcal{A}^{\M_{v,k}}(d_0)$ where $d_0 = 0^{10\Delta}$.  Let $E$ be the event that $\mathcal{A}^{\M_{v,k}}$ on input $d_0$ makes at most $2TK$ calls to $\M$ on any inputs. By Markov's inequality,  $\Pr[E] \geq \frac 1 2$.
Each call to $\M$ made by $\mathcal{A}^{\M_{v,k}}(d_0)$ consists of a single dataset $d$ and a single $j \in [K]$.
When $v, k$ are chosen uniformly at random, the probability that $d \in B_v$ is
\begin{equation*}
     \frac{|B_v|}{2^{10\Delta}} = \frac{{\binom{10\Delta}{\Delta}}}{2^{10\Delta}} \leq (\frac{10e}{1024})^\Delta \leq e^{-3\Delta} \leq \frac 1 {8T}.
\end{equation*}
If the $j$ chosen by $\mathcal{A}^{\M_{v,k}}(d_0)$ is different from $k$, the likelihood of seeing an $(r,1)$ is zero. Thus the likelihood of seeing an $(r,1)$ on a given call to $\M$ is
\ifconf\begin{align*}
    \Pr[\M_{v,k}(d,j) = (r,1)] &\leq \Pr[j=k \wedge d \in B_v] \\&\;\;\;\;\;\;+ \Pr[j=k] \cdot \Pr[\M_v(d) = (r,1) \mid j=k \wedge d \not\in B_v]\\
    &\leq \frac{1}{8TK} + \frac{q}{K} \leq \frac{1}{4TK}.
\end{align*}
\else
\begin{align*}
    \Pr[\M_{v,k}(d,j) = (r,1)] &\leq \Pr[j=k \wedge d \in B_v] + \Pr[j=k] \cdot \Pr[\M_v(d) = (r,1) \mid j=k \wedge d \not\in B_v]\\
    &\leq \frac{1}{8TK} + \frac{q}{K} \leq \frac{1}{4TK}.
\end{align*}
\fi

Conditioned on the event $E$, the probability that $(r, 1)$ is the  output of \emph{any} of the at most $2TK$ runs of $\M$ is at most $\frac{2TK}{4TK} \leq \frac 1 2$. Since $\mathcal{A}^\M$ has never seen a $(r,1)$ output in this case, it follows that $Pr[\mathcal{A}^{\M_{v,k}}(d_0) = (r, 1) \mid E] \leq \frac 1 2$, so that $Pr[\mathcal{A}^{\M}(d_0) = (r', 0) \mid E] \geq \frac 1 2$. It follows that $Pr[\mathcal{A}^{\M}(d_0) = (r', 0)] \geq \frac 1 4$.

The rest of the proof is essentially identical. As $\M_{v,k}(\cdot, j)$ is $2\eps$-DP for each $j$, the privacy property of $\mathcal{A}^{\M_{v,k}}$ implies that it is $(2c\eps, \delta)$-DP. By \cref{prop:ed_dp_probs}, it follows that for small enough $\delta$.
\begin{align*}
     \Pr[\mathcal{A}^{\M_{v,k}}(v) = (r',0)] &\geq e^{-20c\eps\Delta} / 8.
\end{align*}
By the utility property, the left hand side is at most $\gamma$. 
Since $e^{20c\eps\Delta} \leq e^{40c\ln 8T}$, it follows that
 $\gamma  \geq \frac{1}{8(8T)^{40c}}$.
Rearranging, we get the claimed result.
\end{proof}

\section{Upper Bounds Trade-offs}
\label{sec:tradeoffs}
As mentioned earlier, two extreme points on the trade-off between the privacy overhead and the run time overhead are known. We show next the trade-off that is achievable by a simple combination of the basic approaches. We consider the goal of matching the median score, that is the common case for private repetition (labeled \emph{Repetition} below), as well as the goal of matching the $(1-\frac 1 K)$\textit{th} quantile, which is the case for metaselection and hyperparameter tuning (labeled \emph{Metaselection} below). We first state the results one gets from simple repetition.
\begin{theorem}
\label{thm:Atmax}
    Suppose that $\M : D^n \rightarrow (R \times \Re)$ is $(\eps, \delta)$-DP. Then for an integer $T$, the algorithm $\ATmax$ that on input $d$ runs $\M(d)$ $T$ times and releases the output with the highest quality score satisfies
    \begin{itemize}
        \item{Privacy:} ${\ATmax}^{\M} : D^n \rightarrow (R \times \Re)$ is $(T\eps,T\delta)$-DP, as well as $(T\eps^2 + \eps \sqrt{2T\ln \frac 1 {\delta'}}, T\delta + \delta')$-DP, for any $\delta' \in (0,1)$.
        \item{Runtime:} $\mathcal{A}$ on input $d$ makes exactly $T$ calls to $\M$.
        \item{Utility (Repetition):} For any input $d$, if $T\geq \log_2 \gamma$, then $$\Pr[\val({\ATmax}^\M(d)) \geq Median(\val(\M(d)))] \geq 1-\gamma.$$
        \item{Utility (Metaselection):} For any input $d$, and $K> 1$, if $T\geq K \ln \gamma$, then $$\Pr[\val({\ATmax}^\M(d)) \geq q_{1-\tfrac 1 K}(\M(d))] \geq 1-\gamma,$$ where the quantile $q_{1-\tfrac 1 K}$ is such that $\Pr[\M(d) \geq q_{1-\tfrac 1 K}] \geq \frac 1 K$.
    \end{itemize}
\end{theorem}
\begin{proof}
     The privacy follows from simple composition and advanced composition applied and post-processing. The utility and run time analyses are straightforward.
\end{proof}

\noindent The following result was shown by\ifconf Liu and Talwar~\cite{LiuT19}\else~\citet{LiuT19}\fi.
\begin{theorem}
\label{thm:LT}
    Suppose that $\M : D^n \rightarrow (R \times \Re)$ is $\eps$-DP and $\gamma > 0$. Then the algorithm $\A_{LT,\gamma}$ satisfies
    \begin{itemize}
        \item{Privacy:} ${\A_{LT,\gamma}}^{\M} : D^n \rightarrow (R \times \Re)$ is $3\eps$-DP.
        \item{Runtime:} $\A_{LT, \gamma}^{\M}$ on input $d$ makes in expectation  $O(\frac 1 \gamma)$ calls to $\M$.
        \item{Utility (Repetition):} For any input $d$, $\Pr[\val(\A_{LT,\gamma}^\M(d)) \geq Median(\val(\M(d)))] \geq 1-\gamma$.
        \item{Utility (Metaselection):} For any input $d$, and $K >1$, $\Pr[\val({\ATmax}^\M(d)) \geq q_{1-\tfrac 1 K}(\M(d))] \geq 1-K\gamma$, where the quantile $q_{1-\tfrac 1 K}$ is such that $\Pr[\M(d) \geq q_{1-\tfrac 1 K}] \geq \frac 1 K$.
    \end{itemize}
\end{theorem}

\noindent These two algorithms can be combined, by using \cref{thm:LT} to boost the success probability to $1-\gamma^{\frac{1}{c}}$, and then repeating that algorithm similarly to \cref{thm:Atmax} to further boost the success probability to $(1-\gamma)$.
\begin{theorem}
\label{thm:LTmax}
    Suppose that $\M : D^n \rightarrow (R \times \Re)$ is $\eps$-DP and $c>1$ be an integer. Then the algorithm $\A$ that runs $\ATmax$ (for $T=c$) with an ${\A_{LT,\gamma^{\frac 1 c}}^\M}$ oracle satisfies
    \begin{itemize}
        \item{Privacy:} $\A : D^n \rightarrow (R \times \Re)$ is $3c\eps$-DP.
        \item{Runtime:} $\A$ on input $d$ makes in expectation  $O(\frac c {\gamma^{\frac 1 c}})$ calls to $\M$.
        \item{Utility:} For any input $d$, $\Pr[\val(\A(d)) \geq Median(\val(\M(d)))] \geq 1-\gamma$.
    \end{itemize}
\end{theorem}
\begin{proof}
    By \cref{thm:LT},  ${\A_{LT,\gamma^{\frac 1 c}}^\M}$ is $3\eps$-DP. \cref{thm:Atmax} now implies that $\A$ is $3c\eps$-DP. For the chosen parameters, the utility bound in \cref{thm:LT} implies that $\Pr[\val(\A_{LT,\gamma^{\frac 1 c}}^\M(d)) \geq Median(\val(\M(d)))] \geq 1-\gamma^{\frac 1 c}$. Since this algorithm is repeated $c$ times in $\A$, the failure probability gets reduced to $\gamma$. Finally, the run time bound follows by combining the run time results in \cref{thm:LT} and \cref{thm:Atmax}.
\end{proof}

\noindent We note that for $c < \log \gamma^{-1}/\log\log \gamma^{-1}$, $c^c < \gamma^{-1}$ so that this bound of $O(\frac c {\gamma^{\frac 1 c}})$ is $O(\gamma^{-2/c})$. Thus it matches the lower bound in \cref{thm:main_lb_coded} up to constant factors in $c$.

An identical argument yields the following result for metaselection.
\begin{theorem}
\label{thm:LTmax-metaselection}
    Suppose that $\M : D^n \rightarrow (R \times \Re)$ is $\eps$-DP and let $c>1$ be an integer. Then the algorithm $\A$ that runs $\ATmax$ (for $t=c$) with an ${\A_{LT,\gamma^{\frac 1 c}/K}^\M}$ oracle satisfies
    \begin{itemize}
        \item{Privacy:} $\A : D^n \rightarrow (R \times \Re)$ is $3c\eps$-DP.
        \item{Runtime:} $\A$ on input $d$ makes in expectation  $O(\frac {Kc} {\gamma^{\frac 1 c}})$ calls to $\M$.
        \item{Utility:} For any input $d$, $\Pr[\val({\A}^\M(d)) \geq q_{1-\tfrac 1 K}(\M(d))] \geq 1-\gamma$, where the quantile $q_{1-\tfrac 1 K}$ is such that $\Pr[\M(d) \geq q_{1-\tfrac 1 K}] \geq \frac 1 K$.
    \end{itemize}
\end{theorem}
\begin{proof}
    By \cref{thm:LT},  ${\A_{LT,\gamma^{\frac 1 c}/K}^\M}$ is $3\eps$-DP. \cref{thm:Atmax} now implies that $\A$ is $3c\eps$-DP. For the chosen parameters, the utility bound in \cref{thm:LT} implies that $\Pr[\val(\A_{LT,\gamma^{\frac 1 c}}^\M(d)) \geq q_{1-\tfrac 1 K}(\M(d))] \geq 1-K \cdot (\gamma^{\frac 1 c}/K) = 1-\gamma^{\frac 1 c}$. Since this algorithm is repeated $c$ times in $\A$, the failure probability gets reduced to $\gamma$. Finally, the run time bound follows by combining the run time results in \cref{thm:LT} and \cref{thm:Atmax}.
\end{proof}

\noindent
This result improves on $\ATmax$ in terms of privacy cost and gives the points plotted on~\cref{fig:results}(right). As in the case of repetition, this result is tight up to constants in $c$ all the way up to $c = \log \gamma^{-1}/\log\log \gamma^{-1}$.

\section{Conclusions}
\label{sec:conclusions}

Differentially private repetition is a fundamental problem in the design of differentially private algorithms, and private hyperparameter tuning is of great practical importance. Our work shows new trade-offs between privacy overhead and computational overhead, and our main result is a lower bound showing that there is no general algorithm that can do significantly better. Indeed we show that for constant privacy overhead, the computational overhead must be polynomial in $\frac 1 \gamma$ for some private algorithms.

It is natural to ask if there are reasonable restrictions one can place on the inputs or the private algorithms of interest, for which better repetition theorems and/or hyperparameter tuning algorithms are possible. Such beyond-worst-case results are not uncommon in differential privacy~\citep{NissimRS07, DworkL09, ThakurtaS13, AsiD20}. There are some assumptions that hyperparameter tuning algorithms either implicitly~\citep{SnoekLA12} or explicitly~\citep{HazanKY18} make in the non-private setting and one may ask if such assumptions help with better trade-offs for private hyperparameter search. Even absent computational constraints, the $2\times$ or $3\times$ privacy cost overhead in \Alt\ can be large. Once again, while it is unavoidable in the worst-case, one may hope to design algorithms that do better for non-worst-case instances. We leave these important questions for future work.

\bibliography{refs}

\end{document}